\documentclass{elsarticle}

\usepackage{booktabs} 

\usepackage[ruled]{algorithm2e} 

\usepackage[english]{babel}
\usepackage{amsmath}
\usepackage{blindtext}
\usepackage{amsthm}
\usepackage{footnote}
\usepackage[normalem]{ulem}
\usepackage{enumitem}
\usepackage{amsfonts}
\usepackage{cancel}
\usepackage[toc,page]{appendix}

\SetAlFnt{\small}
\SetAlCapFnt{\small}
\SetAlCapNameFnt{\small}
\SetAlCapHSkip{0pt}
\IncMargin{-\parindent}
\newcommand{\E}{\mathbb{E}}

\newcommand{\GG}{\mathcal{G}}
\renewcommand{\P}{\mathbb{P}}
\newtheorem{proposition}{Proposition}
\newtheorem{theorem}{Theorem}
\newtheorem{remark}{Remark}
\newtheorem{corollary}{Corollary}

\usepackage{subcaption}
\usepackage[normalem]{ulem}
\usepackage{url}

\usepackage[english]{babel}
\usepackage{blindtext}

\usepackage{etoolbox}
\makeatletter
\patchcmd{\maketitle}{\@copyrightspace}{}{}{}
\makeatother

\begin{document}





\title{Performance Analysis of Load Balancing Policies with Memory}

\author[1]{Tim Hellemans\corref{cor1}%
\fnref{fn1}}
\ead{timhellemanstim@gmail.com}
\author[1]{Benny Van Houdt}
\cortext[cor1]{Corresponding author}
\address[1]{University of Antwerp, Middelheimlaan 1, Antwerp}

\begin{abstract}
Joining the shortest or least loaded queue among $d$ randomly selected queues are two 
fundamental load balancing policies. Under both policies the dispatcher does not maintain any 
information on the queue length or load of the servers. In this paper we analyze the
performance of these policies when the dispatcher has some memory available to store
the ids of some of the idle servers. We consider methods where the dispatcher
discovers idle servers as well as methods where idle servers inform the dispatcher 
about their state.
 
We focus on large-scale systems and our analysis uses the cavity method. The main insight
provided is that the performance measures obtained via the cavity method for a load
balancing policy {\it with} memory
reduce to the performance measures for the same policy {\it without} memory provided that
the arrival rate is properly scaled. Thus, we can study the performance of load balancers
with memory in the same manner as load balancers without memory.  In particular this entails closed form solutions for joining the shortest or least loaded queue among $d$ randomly selected queues
with memory in case of exponential job sizes. Moreover, we obtain a simple closed form expression for the (scaled) expected waiting time as the system tends towards instability.

We present simulation results that support our belief that the approximation obtained by
the cavity method becomes exact as the number of servers tends to infinity. 
\end{abstract}
\maketitle

\section{Introduction}
Load balancing is often used in large-scale clusters to reduce latency. A simple algorithm, denoted by SQ($d$), exists in assigning incoming jobs to a server that currently holds the least number of jobs out of $d$ randomly selected queues. This is referred to as the \textit{power-of-$d$-choices} algorithm \cite{aghajani2017pde, mitzenmacher2, vvedenskaya3}. Another popular algorithm which has received quite some attention recently exists in assigning an incoming job to the server which is the least loaded amongst $d$ randomly selected queues, i.e.~the server which is able to start working on the job first receives the job. This policy is referred to as LL($d$) and has been studied in e.g.~\cite{hellemans2018power,bramsonLB_QUESTA,Sparrow,mitzenmacher2019supermarket}.

The main objective of this paper is to generalize the analysis of the SQ($d$) and LL($d$) policy to the case where the dispatcher has some (finite or infinite) memory available to store 
the ids of idle servers. These ids may be discovered by either probing servers to check whether they are idle or servers may inform the dispatcher that they became idle. We focus on the performance
of large scale systems and as such make use of the cavity method introduced in \cite{bramsonLB}. The
cavity method relies on the assumption that the queue length (or load) of any finite set of queues
becomes independent as the number of servers tends to infinity, called the {\it ansatz}. 

The ansatz was proven in some particular cases, in~\cite{bramsonLB_QUESTA} it was shown for LL($d$) with general job sizes and SQ($d$) with decreasing hazard rate job sizes. Recently, the ansatz was also proven a variety of load balancing policies which are similar to LL($d$) (see \cite{shneer2020large}). Our objective is not to prove the ansatz for load balancers with memory, but to study these policies using the cavity method. To demonstrate the usefulness
of our analysis we present simulation results which suggest that the cavity method captures the system behavior as the number of servers tends to infinity.

A few papers have previously studied the use of some (bounded) memory at the dispatcher in combination with a power-of-$d$ policy. In \cite{mitzenmacher2002load} the authors study a policy with a memory of size $m$, where at every job arrival $d$ servers are probed and the job is assigned to the server with the smallest number of pending jobs amongst the $d$ probed
servers and $m$ servers in memory. The $m$ servers with the shortest queue amongst the remaining $d+m-1$ servers form the memory for the next job arrival. In \cite{gamarnik2020lower} the authors study the amount of memory resp.~probes required in order to obtain vanishing queueing delay. In \cite{van2019hyper, anselmi2020power}  policies are studied where the dispatcher maintains an upper bound on the queue length
of each server and dispatches jobs based on these upper bounds.

The main insight obtained in this paper is that studying a load balancing policy with memory
using the cavity method, corresponds to studying the same load balancing policy without memory 
if we scale down the arrival rate in a proper manner (see also Theorem \ref{thm:response_SQd}, \ref{thm:response_SQd_PHD}, \ref{thm:response_SQd_gen} and \ref{thm:response_LLd}). For the LL($d$) policy, we do not impose any restrictions on the job size distribution. For the SQ($d$) policy, we initially restrict our attention to exponential job sizes and then generalize our main result to phase type and general job size
distributions.

As a by-product, our results allow us to study the Join-Idle-Queue policy (denoted by JIQ) 
with finite memory. JIQ  exists in keeping track of the ids of the idle queues and assigning 
incoming jobs to an idle queue whenever there is an idle server in memory and simply assigning it to a random server otherwise. This policy has vanishing delays when the number of servers tends to infinity in case of infinite memory \cite{lu2011join, stolyar1, foss2017large, braverman2018steady}. 

Apart from the cavity method analysis, we additionaly present explicit results for
the heavy traffic limit by relying on the framework in \cite{hellemans2020heavy} that allows one to compute the limit $\lim_{\lambda \rightarrow 1^-} - \frac{\mathbb{E}[R_\lambda]}{\log(1-\lambda)}$ for load balancing policies with exponential job sizes. 
We show that (unsurprisingly) for most memory schemes, the heavy traffic limit remains unchanged when we add memory at the dispatcher. However, when the dispatcher has a memory of size $A$ and servers inform the dispatcher when they become idle, the heavy traffic limit is multiplied by $\frac{1}{A+1}$ for both the LL($d$) and SQ($d$) policy. In particular with a memory size of $1$, the response time under heavy traffic is halved compared to having no memory at the dispatcher.

Finally, we analyze the low traffic limit in case of exponential job sizes. In particular, we take a closer look at the ratio of the mean waiting time for two different load balancing policies as the load tends to zero, for which we find a simple closed form solution.

The paper is structured as follows. In Section \ref{sec:model_description}, the model is introduced and we shortly review previously obtained results for SQ($d$) and LL($d$). In Section \ref{sec:examples} we present four approaches which make use of memory at the dispatcher and show how to obtain the probability that the memory is empty for each of these methods. Next we present our major analytical tool, the queue at the cavity in Section \ref{sec:cavity} and we describe how it is defined for the memory dependent LL($d$) and SQ($d$) policy. We carry out the analysis of the queue at the cavity in Section \ref{sec:analysis}. Our analysis is verified by means of simulation in Section \ref{sec:finAcc}. In Section \ref{sec:num_examples} we show how our results may be used for numerical experimentation by studying one specific setting. In Section \ref{sec:heavy} we study the heavy traffic limit, while in Section \ref{sec:low} the low traffic limit is considered. We conclude the paper in Section \ref{sec:conclusions} and
discuss possible future work.

All code used to generate Table \ref{table:fin_acc} and Figure \ref{fig:num_exp_SQd} can be found at\newline \textit{https://github.com/THellemans/memoryDependentLB.git}.

\section{Model Description}\label{sec:model_description}
We consider a system consisting of $N$ identical servers (with $N$ large). There is some central dispatcher to which jobs arrive according to a Poisson($\lambda N$) process. The dispatcher has some
(finite or infinite) memory available to store ids of idle servers. When a job arrives and the dispatcher has the id(s) of some idle server(s) in its memory, the job is dispatched to 
a random server, the id of which is in memory. If the dispatcher's memory is empty, $d$ servers are chosen at random and the job is either send to the server with the shortest queue (SQ($d$), see Section \ref{sec:SQd_explained}) or to the server with the least amount of work (LL($d$), see Section \ref{sec:LLd_explained}). Setting $d=1$ yields the JIQ policy where the job is simply routed arbitrarily whenever there are no idle servers known by the dispatcher. Before we proceed we provide some further details on the classic SQ($d$) and LL($d$) policy.


\subsection{Classic SQ($d$)} \label{sec:SQd_explained}
The SQ($d$) policy was first introduced in \cite{mitzenmacher2, vvedenskaya3} for a system with
Poisson($\lambda N$) arrivals and exponential job sizes (with mean $1/\mu$). Whenever a job arrives, 
$d$ servers are probed at random and the incoming job is routed to the probed server with the least number of jobs in its queue. It was shown (see \cite{mitzenmacher2}) that in the limit as $N \rightarrow \infty$
the system behavior converges to the solution of the following set of ODEs:
$$
\frac{d}{dt} u_k(t)
=
\lambda (u_{k-1}(t)^d-u_k(t)^d) - \mu(u_k(t)-u_{k+1}(t)),
$$
where we denote by $u_k(t)$ the probability that, at time $t$, an arbitrary server has at least $k$ jobs in its queue and $u_0(t)=1$. This set of ODEs also corresponds to applying the cavity method to the SQ($d$) policy. The fixed point of this set of ODEs obeys a simple recursive formula:
\begin{equation}\label{eq:SQd_vanilla}
\mu u_{k+1}= \lambda u_k^d,
\end{equation}
which has the closed form solution $u_k = \rho^{\frac{d^k-1}{d-1}}$
with $\rho = \lambda/\mu$. In particular one obtains from Little's Law the closed form solution of the expected response time:
\begin{equation}\label{eq:ER_SQd_vanilla}
\E[R]= \frac{1}{\lambda} \sum_{k=1}^\infty \rho^{\frac{d^k-1}{d-1}}.
\end{equation}

\subsection{Classic LL($d$)} \label{sec:LLd_explained}
The LL($d$) policy was analyzed in \cite{hellemans2018power} for a system with arbitrary job sizes
with mean $\E[G]$
using the cavity method. Whenever a job arrives, $d$ queues are probed and the job is sent to the queue which has the least amount of work left. This means that the job joins the queue at which its service can start the soonest. In practice this can be implemented through late binding, see also \cite{Sparrow}. 
Let $\bar F(w)$ denote the equilibrium probability that an arbitrary queue has at least $w$ work left
using the cavity method. 
It is shown in \cite{hellemans2018power} that $\bar F(w)$ satisfies the fixed point equation:
\begin{equation} \label{eq:FbarW_classicSQd_IDE}
\bar F(w)
=
\rho - \lambda \int_0^w (1-\bar F(u)^d) \bar G(w-u) \, du,
\end{equation}
with $\rho = \lambda \E[G]$ and $\bar G(w-u)$ the probability that a job has a size greater than $w-u$.
This fixed point equation can alternatively be written as the following Integro Differential Equation (IDE):
$$
\bar F'(w)
=
-\lambda \bigg[
\bar G(w) - \bar F(w)^d + \int_0^w \bar F(u)^d g(w-u) \, du
\bigg],
$$
with $g$ the density function of the job size distribution.
Both have the boundary condition $\bar F(0) = \rho$. Moreover, this equation has a closed form solution in case of exponential job sizes (with mean $1/\mu$):
\begin{equation} \label{eq:FbarW_classicSQd_closed}
\bar F(w)=(\rho + (\rho^{1-d}-\rho) e^{(d-1)w})^{\frac{1}{1-d}}.
\end{equation}
In particular, one obtains a closed form solution for the expected response time:
\begin{equation} \label{eq:ER_LLd_vanilla}
\E[R]= \frac{1}{\lambda} \sum_{n=0}^\infty \frac{\rho^{dn+1}}{1+n\cdot (d-1)}.
\end{equation}

\section{Discovering idle servers} \label{sec:examples}
We now discuss a number of approaches for the dispatcher to discover ids of idle servers.
In the first few approaches the dispatcher discovers idle servers by probing, while in the
last approach the idle servers identify themselves to the dispatcher.  Note that as the amount of incoming work per server per unit of time is equal to $\rho<1$, no work is replicated, and all servers are identical, it follows that the steady state probability that a server is busy is given by 
$\rho$.\footnote{For SQ($d$) with exponential job sizes this is shown explicitly in the proof of Theorem \ref{thm:closed_SQd}, while for SQ($d$) with general job sizes the proof is carried out in Proposition \ref{prop:u1_eq_rho_SQd_gen}. For LL($d$) this easily follows from integrating both sides of \eqref{eq:IDE_LLd}.}

\subsection{Interrupted probing (IP)}
In the first approach, called {\it interrupted} probing (IP), the dispatcher probes $d$ servers when its memory is empty upon a job arrival. If there are $k \geq 1$ idle servers among the $d$ probed servers, 
it sends the incoming job to one of the idle servers and stores ids of the $k-1$ other
servers in memory. These $k-1$ ids are then used for the subsequent $k-1$ arrivals. 
Thus for these $k-1$ arrivals, the dispatcher does not probe any servers. 
As $\rho$ is the steady state probability that a server is busy, we can find the probability
$\pi_0$ of having no ids in memory when a new job arrives by looking at the Markov chain with state space ${0,\ldots,d-1}$ and transition probability matrix $M(\rho)$:
\begin{align*}
M(\rho)_{0,0}&= \rho^d + \binom{d}{1} \rho^{d-1} (1-\rho),\\
M(\rho)_{0,\ell}&= \binom{d}{\ell+1} \rho^{d-1-\ell} (1-\rho)^{\ell+1},\\
M(\rho)_{k,k-1}&=1,\\
\end{align*}
and $M(\rho)_{k,\ell}=0$ otherwise.

As only the first row is non-trivial, it is not hard to check that $\pi=(\pi_0,\ldots,\pi_{d-1})$
given by:
$$
\pi_k=\pi_0 \left[
1- \sum_{j=0}^k \binom{d}{j} \rho^{d-j} (1-\rho)^j
\right],
$$
for $k \geq 1$ is an invariant vector of $M(\rho)$.
From the requirement $\sum_{k=0}^{d-1} \pi_k=1$ it then follows that 
\begin{align}\label{eq:pi0IP}
\pi_0 = \frac{1}{\rho^d + (1-\rho)d}.
\end{align}

The number of probes used per arrival is clearly given by $\pi_0 d$ which equals
$$
\frac{1}{1-\rho+\frac{\rho^d}{d}}.
$$
The main advantage of the IP approach is that it uses far less than $d$ probes per arrival 
when $\rho$ is not too large (see also Figure \ref{fig:num_probes}).

\subsection{Continuous probing (CP)} \label{sec:probe_lot_memory}
This approach is similar to the IP approach, except that whenever we use a server id
from memory for a job arrival, the dispatcher still probes $d$ random servers. The ids
of the $d$ servers that are idle are subsequently added to memory. We assume that 
the available memory is unlimited. 

In order to determine the probability $\pi_0$ of having a server id in memory,
we need to study a Markov chain on an infinite state space.
Its transition probability matrix $M(\rho)$ has the following form:
\begin{align*}
M(\rho)_{0,0}&= \rho^d + \binom{d}{1} \rho^{d-1} (1-\rho),\\
M(\rho)_{0,\ell}&= \binom{d}{\ell+1} \rho^{d-1-\ell} (1-\rho)^{\ell+1},
\end{align*}
for $1\leq \ell \leq d-1$. For $k\geq 1$, we have 
\begin{align*}
M(\rho)_{k,k-1+\ell}
&=\binom{d}{\ell} \rho^{d-\ell} (1-\rho)^{\ell},
\end{align*}
for $k-1\leq \ell <d+k$, and $M(\rho)_{k,\ell}=0$ otherwise.
First note that if $d>\frac{1}{1-\rho}$, this Markov chain is transient
as the drift in state $k > 0$ is given by $d(1-\rho)-1$,
meaning after some point in time the chain never returns to state zero and all
incoming arrivals can be assigned to an idle server.
When $d < \frac{1}{1-\rho}$, the chain is positive recurrent and we need to 
determine $\pi_0 < 1$. A similar observation was made in \cite{van2019hyper, anselmi2020power}.

The average time the memory remains empty is equal to:
$$
\frac{1}{1-M(\rho)_{0,0}}=\frac{1}{1- \rho^d - d \rho^{d-1} (1-\rho)}.
$$
Furthermore, when the memory becomes non-empty, the length that it remains non-empty
depends on the number of server ids that are placed into memory. More specifically let $\E[T_{k,0}]$ denote 
the expected first return time to $0$ given that the chain starts in state $k > 0$,
then:
$$
\E[T_{k,0}] = k \E[T_{1,0}],
$$
and the mean time that the Markov chain stays away from state $0$ given that
it just made a jump from state $0$ to some state $k > 0$ is given by $\E[X_0]\E[T_{1,0}]$,
where $\E[X_0]$ is one less than the mean number of idle servers among $d$ servers given that
at least $2$ are idle. It is not hard to see that 
$$\E[X_0]=\frac{d(1-\rho)-(1-\rho^d)}{1- \rho^d - d \rho^{d-1} (1-\rho)}.$$
Further, $\E[T_{1,0}]=\frac{1}{1-d(1-\rho)}$ as $\E[T_{1,0}] = 1 + d(1-\rho) \E[T_{1,0}]$. 
Putting this together we obtain
\begin{align}
\pi_0&=\frac{1-d(1-\rho)}{\rho^d},\label{eq:pi_0_lots_of_mem}
\end{align}
when $d < \frac{1}{1-\rho}$.
Note that the CP approach uses $d$ probes per arrival.

\subsection{Bounded Continuous Probing (BCP)}
This approach is identical to the CP approach, but with finite memory size $A$. 
Hence the transition matrix $M(\rho)$ is of size $A+1$ and its transitions are the same as in Section \ref{sec:probe_lot_memory}, except that any transition from a state $k \leq A$
to a state $\ell > A$ becomes a transition to state $A$. In particular for $k > A-d+1$ we have:
$$
M(\rho)_{k,A}
=
\sum_{j=A-k+1}^d \binom{d}{j} \rho^{d-j} (1-\rho)^{j},
$$
and for all other values, $M(\rho)$ coincides with the expressions given in Section \ref{sec:probe_lot_memory}. This Markov chain does not appear to have a simple closed form solution for arbitrary values of $d$, however for $d=2$ one finds: 
$$
\pi_0=\frac{1-\left(\frac{1-\rho}{\rho}\right)^2}{1-\left(\frac{1-\rho}{\rho}\right)^{2(A+1)}}.
$$
For $d>2$ a simple numerical scheme can be used to compute $\pi_0$. 
Note that this approach uses $d$ probes per arrival unless the dispatcher sends the probes one at a
time and stops probing when the memory is full.

\subsection{Other probing schemes}
In this section we present a result that applies to any scheme where the dispatcher 
discovers idle servers by probing and any idle server that is discovered is stored in memory.
Thus the result only applies to BCP if the probes are sent one at a time.
 
\begin{proposition}\label{prop:nr_probes}
Assume all discovered idle servers are stored in memory. Then for 
any LL($d$)/SQ($d$) memory based policy, the average number of probes used 
per arrival is given by:
\begin{equation}
\frac{1-\pi_0 \rho^d}{1-\rho}.
\label{eq:nr_probes}
\end{equation}
\end{proposition}
\begin{proof}
If we think of the probes being transmitted one at a time
and assigning the job as soon as an idle server is discovered, the dispatcher uses on average
$\sum_{k=0}^{d-1} \rho^k$ probes for any job arrival that occurs when the memory is
empty. Further, for any arrival that uses an id in memory, an average of $1/(1-\rho)$ probes was
used to discover that id. Hence, the average number of probes transmitted can be written as:
$$
\pi_0 \frac{1-\rho^d}{1-\rho}+ (1-\pi_0) \frac{1}{1-\rho}.
$$
\end{proof}

The above result indicates that for any such policy for which we either know the 
average number of probed queues (as for CP) or can express this using $\pi_0$
(as for IP), we immediately obtain $\pi_0$. 
As the CP policy sends $d$ probes per arrival and the IP policy $\pi_0 d$,  
Proposition  \ref{prop:nr_probes} yields 
\eqref{eq:pi0IP} and \eqref{eq:pi_0_lots_of_mem} without the need to analyze a 
Markov chain.


\subsection{Idle Server Messaging (ISM)}\label{sec:example_JIQ}
In this scheme the dispatcher does not probe to discover idle servers, instead a server 
notifies the dispatcher whenever it becomes idle. In case of infinite memory, the dispatcher
knows all idle server ids at all times and the system reduces to the JIQ policy
when $d=1$. Our interest lies mostly in knowing what happens when the memory size is finite
and the job is assigned to the shortest of $d$ queues whenever the memory is empty
when a job arrives.

If we denote $A$ as the number of ids that can be stored in memory, we show that $\pi_0$ 
is given by
\begin{equation}\label{eq:pi_0_JIQ}
\pi_0
=
\frac{1-(1-\rho^d)^{\frac{1}{A+1}}}{\rho^d}.
\end{equation}
For SQ($d$) this is shown in Proposition \ref{prop:pi_0_JIQ_SQd_gen}, while for LL($d$) this is presented in Proposition \ref{prop:pi_0_JIQ_LLd}. In particular, this result entails that $\pi_0$ is insensitive to the job size distribution for SQ($d$) and LL($d$).

If we assume that the $d$ probes are transmitted one at a time when memory is empty and the dispatcher stops probing as soon as an idle server is discovered, the number of probes and messages transmitted by the
dispatchers and servers per job arrival can be expressed as:
\[\pi_0\frac{1-\rho^d}{1-\rho}+(1-\pi_0 \rho^d),\]
where the first term corresponds to the number of probes send per arrival by the
dispatcher and the second correspond to the number of server messages per arrival
(which is equal to the probability that a job is assigned to an idle server).



\section{Description of the queue at the cavity}  \label{sec:cavity}
Our analysis is based on the {\it queue at the cavity} method which was introduced  in \cite{bramsonLB} to analyze load balancing systems. The key idea is to focus on the evolution of a single tagged queue,
referred to as the queue at the cavity, and 
to assume that all other queues have the same queue length (or workload) distribution at any time $t$.
Moreover the queue length   (or workload) of any finite set of queues is 
assumed to be independent at any time $t$. We first explain the approach in a system without memory
and then indicate how to adapt it to incorporate memory.

In a system without memory, the queue at the cavity experiences potential arrivals at rate $\lambda d$
as this is the rate at which a tagged queue is selected as one of the $d$ randomly selected queues. 
If  a potential arrival occurs at time $t$, $d-1$ i.i.d.~random variables are initialized which have the same queue length (or workload) distribution as the queue at the cavity at time $t$. The potential arrival becomes an actual arrival
if the queue at the cavity has the shortest queue (or smallest workload) amongst these $d$ values (where ties are broken at random).
For SQ($d$) with exponential job sizes with mean $1/\mu$ the queue length of the queue at the cavity decreases at a constant rate equal to $\mu$ in between potential arrivals, while for LL($d$) the workload
decreases linearly at rate $1$ when larger than zero. For Phase Type distributed job sizes, one needs to include the phase of the job at the head of the queue, while for general job sizes we need to include the work left for the job at the head of the queue.

To incorporate memory into the cavity method we note that the state of the memory (that is, the number of ids that it contains) evolves at a faster time scale than the fraction of queues with a
certain queue length (or workload). As such the state of the memory at time $t$ is 
given by the steady state $\pi(t)$ of the discrete time Markov chain with transition
matrix $M(\rho(t))$ that captures the 
evolution of the memory, where $\rho(t)$ is the fraction of busy servers at time $t$
(see Section \ref{sec:examples} for some examples with $\rho(t)=\rho$). For more details on the concept of the time-scale separation we employ, we refer the reader to \cite{benaim2008class}.

Let $\pi_0(t)$, the first entry of $\pi(t)$, represent the probability that the memory is empty at time $t$. 
We modify the queue at the cavity by decreasing the potential arrival rate to the queue at the cavity to $\lambda d \pi_0(t)$, i.e.~potential arrivals only occur when there is no empty queue to join in memory. These potential arrivals are then dealt with in the exact same manner as in the setting without memory. When the queue at the cavity is empty, we assume that on top of the potential arrival rate of $\lambda d \pi_0(t)$, we have an effective arrival rate of $\lambda \frac{1-\pi_0(t)}{1-\rho(t)}$. 
The latter arrival rate can be interpreted as follows: jobs arrive at rate $\lambda N$,
with probability $(1-\pi_0(t))$ such a job is assigned to a queue in memory and with
probability $1/((1-\rho(t))N)$ the queue at the cavity gets the job as it is one of
the $(1-\rho(t))N$ idle servers at time $t$.

In the next section we study the
cavity process of SQ($d$) and LL($d$) with memory in detail. 
We assume job sizes have some general distribution with 
probability density function (pdf) $g$, 
cumulative distribution function (cdf) $G$ and complementary cdf (ccdf) $\bar G$. For a random variable with cdf $H$ we let $\E[H]$ denote its mean. Let $\mu = \frac{1}{\E[G]}$ denote the mean service rate and note that we have for the system load: $\rho = \lambda \cdot \E[G]$. Furthermore we let $\GG$ denote a generic random variable with distribution $G$. We will sometimes assume that $\GG$ is an exponential random variable. Furthermore, for LL($d$) we denote by $f, F$ and $\bar F$ the pdf, cdf and ccdf of the workload distribution of the queue at the cavity in equilibrium (note that we have $\bar F(0) = \rho$). For SQ($d$) with exponential job sizes we denote by $u_k$ the equilibrium probability that the queue at the cavity has $k$ or more jobs (with $u_0=1$ and $u_1=\rho$).

\section{Analysis of the queue at the cavity}\label{sec:analysis}
We now analyze the queue at the cavity described in the previous section for SQ($d$) and LL($d$).
Note that the results presented in this section apply to any of the memory schemes discussed in
Section \ref{sec:examples}. To obtain results for a specific memory scheme one simply replaces $\pi_0$
by the appropriate expression. We show that the equilibrium queue length and workload distribution
of SQ($d$) and LL($d$) with memory, respectively, have exactly the same form as in the same setting 
without memory if we replace $\lambda$ by $\lambda \pi_0^{1/d}$ and divide by $\pi_0^{1/d}$.
With respect to the response time distribution, we show that the system with memory 
and arrival rate $\lambda$ has the same response time distribution as the system without memory and arrival rate $\lambda \pi_0^{1/d}$.


\subsection{SQ($d$)}
In this section we develop the analysis of the queue at the cavity for the SQ($d$) policy, we start by assuming job sizes are exponential and subsequently we consider Phase Type and general job sizes.
\subsubsection{Exponential Job Sizes}
We start by describing the transient behaviour of the queue at the cavity for SQ($d$):
\begin{proposition}\label{prop:trans} 
Consider the SQ($d$) policy with memory, exponential job sizes with mean $1/\mu$ and
arrival rate $\lambda < \mu$.
Let $u_k(t)$ be the probability that the queue at the cavity has $k$ or more jobs at time $t$, then
\begin{align}
\frac{d}{dt} u_k(t)
&=
\lambda \pi_0(t) (u_{k-1}(t)^d - u_k(t)^d)
-
\mu (u_k(t)-u_{k+1}(t)),
  \label{eq:transient_memory_SQd1}\\
\frac{d}{dt} u_1(t)
&=
\lambda \pi_0(t) (u_{0}(t)^d - u_1(t)^d) +
\lambda (1-\pi_0(t))\
- \mu (u_1(t) - u_{2}(t)), \label{eq:transient_memory_SQd2}
\end{align}
for $k \geq 2$ and $u_0(t)=1$.
\end{proposition}
\begin{proof}
Let $\Delta > 0$ be arbitrary, we first assume that $k\geq 2$ and consider the cases in which the queue at the cavity may have $k$ or more jobs at time $t+\Delta$. First, it may have exactly $k$ jobs at time $t$ and no departures occur in $[t,t+\Delta]$, this occurs with probability:
\begin{equation}
Q_{1,k}=(1-\mu \Delta) (u_k(t)-u_{k+1}(t)) + o(\Delta). \label{eq:transient_uk_proof_1.1}
\end{equation}
It may also have $k+1$ or more jobs at time $t$, and at most $1$ departure occurs in $[t,t+\Delta]$:
\begin{equation}
Q_{2,k}=u_{k+1}(t) + o(\Delta). \label{eq:transient_uk_proof_1.2}
\end{equation}
A third possibility is that it had exactly $k-1$ jobs at time $t$ and exactly one arrival occurs in $[t,t+\Delta]$ which joined the queue at the cavity, this occurs with probability:
\begin{align}
Q_{3,k}&=\lambda d  \left(\int_0^\Delta \pi_0(t+\delta) d\delta \right)  (u_{k-1}(t)-u_k(t)) \\
&\sum_{j=0}^{d-1} \frac{1}{j+1	} \binom{d-1}{j} (u_{k-1}(t)-u_k(t))^j \cdot u_k(t)^{d-j} +o(\Delta)\nonumber\\
&= \lambda \left(\int_0^\Delta \pi_0(t+\delta) d\delta \right)  (u_{k-1}(t)^d - u_k(t)^d)+o(\Delta).
\label{eq:transient_uk_proof_2}
\end{align}
We now obtain:
$$
u_{k}(t+\Delta)= Q_{1,k}+Q_{2,k}+Q_{3,k},
$$
subtracting $u_k(t)$ on both sides, dividing by $\Delta$ and taking the limit $\Delta \rightarrow 0$ yields \eqref{eq:transient_memory_SQd1}. For \eqref{eq:transient_memory_SQd2}, one needs to consider the same $Q_{1,k},Q_{2,k}$ and $Q_{3,k}$ as for $k\geq 2$ for the case of potential arrivals. There is however an additional term for the case where the queue at the cavity is empty at time $t$ and it experiences an arrival due to the memory induced arrival rate, this yields:
$$
Q_{4,1}
=
\lambda \left( \int_0^\Delta \frac{1-\pi_0(t+\delta)}{u_0(t+\delta)-u_1(t+\delta)} d\delta \right) (u_{0}(t)-u_1(t)) + o(\Delta),
$$
one then obtains $u_1(t+\Delta)=Q_{1,1}+Q_{2,1}+Q_{3,1}+Q_{4,1}$, subtracting $u_1(t)$, dividing both sides by $\Delta$ and taking the limit $\Delta \rightarrow 0$ yields \eqref{eq:transient_memory_SQd2}. Finally the last equation $u_0(t)=1$ is trivial by the definition of $u_0(t)$.
\end{proof}

From the transient regime, we are able to deduce the equilibrium workload distribution:
\begin{theorem}\label{thm:closed_SQd}
Consider the SQ($d$) policy with memory, exponential job sizes with mean $1/\mu$ and
arrival rate $\lambda < \mu$.
Let $u_k$ be the equilibrium probability that the queue at the cavity has $k$ or more jobs, then
\begin{equation}
u_k = \rho^{\frac{d^{k} - 1}{d-1}} \cdot \pi_0^{\frac{d^{k-1}-1}{d-1}}
= (\rho \pi_0^{1/d})^{\frac{d^{k} - 1}{d-1}}/\pi_0^{1/d},\label{eq:closed_SQd}
\end{equation}
for $k \geq 1$ and $\rho=\lambda/\mu$.
\end{theorem}
\begin{proof}
Taking the limit of $t\rightarrow \infty$ in (\ref{eq:transient_memory_SQd1}-\ref{eq:transient_memory_SQd2}) we find that the following holds:
\begin{align*}
0 &= \lambda \pi_0 (u_0^d - u_1^d)+ \frac{\lambda (1-\pi_0)}{1-\rho} \cdot (u_0-u_1)-\mu \cdot (u_1 - u_2),\\
0 &= \lambda \pi_0 ( u_{k-1}^d - u_{k}^d ) - \mu \cdot (u_{k} - u_{k+1}),
\end{align*}
for $k \geq 2$. Summing all of these equations yields $u_1 = \rho$, while
taking the sum for $k\geq j$ implies that $u_j = \lambda \pi_0 u_{j-1}^d$ for
$j \geq 2$. This simple recurrence relation has \eqref{eq:closed_SQd} as its unique solution.
\end{proof}
Comparing \eqref{eq:closed_SQd} with the solution of \eqref{eq:SQd_vanilla}, we see that $u_k$
is identical as in the setting without memory if we replace $\rho$ by $\rho \pi_0^{1/d}$ and
divide by $\pi_0^{1/d}$ (even for $k = 1$).

\begin{theorem} \label{thm:response_SQd}
Let $0< \lambda < \mu$ be arbitrary and $R$ the response time of the SQ($d$) policy with 
memory, exponential job sizes with mean $1/\mu$ and arrival rate $\lambda$. Further, let $\tilde R$ denote the response time for the same system without memory, but with arrival rate $\lambda \pi_0^{1/d}$, then $\tilde R$ and $R$ have the same distribution.
\end{theorem}
\begin{proof}
Let us denote by $u_k$  and $v_k$ the probability that the queue at the cavity has at least $k$ jobs for the system with and without memory, respectively. 
We have $u_k=\pi_0^{-1/d} \cdot v_k$, for $k \geq 1$ and $u_0 = v_0 = 1$. Let $\bar F_X$ be the ccdf of $X$, then
\begin{align*}
\bar F_{R}(w)
&=(1-\pi_0) e^{-\mu w} + \pi_0 \sum_{k=0}^\infty (u_k^d-u_{k+1}^d) \sum_{n=0}^k \frac{w^n}{n!} e^{-\mu w},
\end{align*}
as with probability $(1-\pi_0)$ the job joins an idle queue from memory (meaning the response time is
simply exponential) and with
probability $\pi_0 (u_k^d-u_{k+1}^d)$ the job joins a queue with length $k$
(yielding an Erlang $k+1$ response time). Exchanging the order of the sums 
and using $\pi_0 u_k^d= v_k^d$, for $k \geq 1$, implies that
\begin{align*}
\bar F_{R}(w)
&=(1-\pi_0) e^{-\mu w} + \sum_{n=1}^\infty \frac{w^n}{n!} e^{-\mu w} v_n^d + \pi_0 e^{-\mu w} =\sum_{n=0}^\infty \frac{w^n}{n!} e^{-\mu w} v_n^d.
\end{align*}
Similarly,
\begin{align*}
\bar F_{\tilde R}(w)
&=\sum_{k=0}^\infty (v_k^d-v_{k+1}^d) \sum_{n=0}^k \frac{w^n}{n!} e^{-\mu w} =\sum_{n=0}^\infty \frac{w^n}{n!} e^{-\mu w} v_n^d.
\end{align*}
\end{proof}

\subsubsection{Phase Type Job Sizes}\label{sec:SQd_PH}
Phase Type (PH) distributions consist of all distributions which have a modulating finite background Markov chain (see also \cite{latouche1999introduction}). They form a broad spectrum of distributions as any positive valued distribution can be approximated arbitrarily close by a PH distribution. Moreover, various fitting tools are available online for PH distributions (e.g.~\cite{kriege2014ph,panchenko1}). A PH distribution with $\bar G(0) =1$ is fully characterized by a stochastic vector $\alpha=(\alpha_i)_{i=1}^n$ and a subgenerator matrix $A=(a_{i,j})_{i,j=1}^n$ such that $\bar G(w)=\alpha e^{Aw} \textbf{1}$, where $\textbf{1}$ is a column vector of ones.

We find that the result found in Theorem \ref{thm:response_SQd} generalizes to the case of PH distributed job sizes.
\begin{theorem} \label{thm:response_SQd_PHD}
Let $0<\lambda<\mu$ (with $1/\mu$ the mean of the job size distribution) be arbitrary and $R$ the response time for a memory dependent version of the SQ($d$) policy with PH distributed job sizes with parameters $(\alpha, A)$. Further, let $\tilde R$ denote the response time for the
classic SQ($d$) policy with the same job size distribution and arrival rate $\lambda \pi_0^{1/d}$, then $R$ and $\tilde R$ have the same distribution.
\end{theorem}
\begin{proof}
Let us denote by $u_{k,j}(t)$ resp.~$v_{k,j}(t)$ the probability that, at time $t$, the queue at the cavity has at least $k$ jobs and the job at the head of the queue is in phase $j$ for the memory dependent scheme resp.~the memory independent scheme. Furthermore let $u_{k,j}$ and $v_{k,j}$ denote the limit of $t\rightarrow\infty$ for these values. We first show that $u_{k,j}=\pi_0^{-1/d} \cdot v_{k,j}$. Throughout we let $\nu = -A \textbf{1}$ (with \textbf{1} a vector consisting of only ones). For $v_{k,j}$ we find by an analogous reasoning as in \cite{van2015performance} that for $k\geq 2$:
\begin{align}
\frac{d}{dt} v_{k,j}(t)
&= \lambda \pi_0(t)^{1/d} \frac{v_{k-1,j}(t)-v_{k,j}(t)}{v_{k-1}(t)-v_k(t)} (v_{k-1}(t)^d-v_k(t)^d \nonumber\\
&+ \sum_{j'} \left( v_{k,j'}(t) A_{j',j} + v_{k+1,j'}(t) \nu_{j'} \alpha_j\right), \label{eq:SQd_PH_proof1}
\end{align}
where $v_k(t)$ denotes $\sum_j v_{k,j}(t)$ (further on, we also use this notation for $v_k, u_k(t)$ and $u_k$). For $k=1$ we find:
\begin{align}
\frac{d}{dt} v_{1,j}(t) &= \alpha_j \lambda \pi_0(t)^{1/d} (1-v_1(t)^d) + \sum_{j'} \left( v_{1,j'}(t) A_{j',j} + v_{2,j'}(t) \nu_{j'} \alpha_j\right).\label{eq:SQd_PH_proof2}
\end{align}
Taking the limit of $t$ to infinity and multiplying by $\pi_0^{-1/d}$ we find that \eqref{eq:SQd_PH_proof1} yields for the equilibrium distribution (with $k\geq 2$):
\begin{align}
0&= \pi_0 \lambda \frac{(\pi_0^{-1/d} v_{k-1,j})- (\pi_0^{-1/d} v_{k,j})}{(\pi_0^{-1/d} v_{k-1})-(\pi_0^{-1/d} v_{k})} \cdot \left( (\pi_0^{-1/d} v_{k-1})^d - (\pi_0^{-1/d} v_{k})^d \right)\nonumber\\
&+ \sum_{j'} (\pi_0^{-1/d} v_{k,j'}) A_{j',j} + (\pi_0^{-1/d} v_{k+1,j'}) \nu_{j'} \alpha_j.\label{eq:SQd_PH_proof3}
\end{align}
while for $k=1$ one may compute from \eqref{eq:SQd_PH_proof2}:
\begin{align}
0
&=
\alpha_j \lambda \left(1- \pi_0 (\pi_0^{-1/d}v_1)^d \right)\nonumber\\
& + \sum_{j'} \bigg( (\pi_0^{-1/d} v_{1,j'}) A_{j',j} + (\pi_0^{-1/d} v_{2,j'})\nu_{j'} \alpha_j \bigg)\label{eq:SQd_PH_proof4}
\end{align}
For $(u_{k,j}(t))$ with $k\geq 2$, we find the same ODE as \eqref{eq:SQd_PH_proof1} but with $\lambda \pi_0(t)$ rather than $\lambda \pi_0^{1/d}(t)$. Taking the limit $t\rightarrow \infty$ it is not hard to see that $u_{k,j}$ satisfies \eqref{eq:SQd_PH_proof3} with $\pi_0^{-1/d} v_{k,j}$ replaced by $u_{k,j}$. Furthermore for $u_{1,j}(t)$ we find (similar to Proposition \ref{prop:trans}):
\begin{align*}
\frac{d}{dt} u_{1,j}(t)
&= \lambda \alpha_j \pi_0(t) (1-u_1(t)^d) + \lambda \alpha_j (1-\pi_0(t))\nonumber \\
& +\sum_{j'} (u_{1,j'}(t) A_{j',j} + u_{2,j'}(t) \nu_{j'} \alpha_j).
\end{align*}
Taking $t\rightarrow \infty$ it is not hard to see how this equation for $u_{k,j}$ reduces to \eqref{eq:SQd_PH_proof4} with $\pi_0^{-1/d} v_{k,j}$ replaced by $u_{k,j}$. This shows that we indeed have for all $k$ and $j$ that $u_{k,j}=\pi_0^{-1/d} v_{k,j}$.

For the response time distribution we denote by $X_{k,j}$ the response time of a job
that joins a queue with length $k$ in phase $j$. We find for the memory dependent policy:
\begin{align*}
\bar F_{R}(w)
&= (1-\pi_0) \bar G(w) + \pi_0 \bigg( (1-u_1^d) \bar G(w)\\
& + \sum_{k=1}^\infty \sum_j \frac{u_{k,j}-u_{k+1,j}}{u_k-u_{k+1}} \cdot (u_k^d-u_{k+1}^d) \P\{X_{k,j} > w\} \bigg)\\
&=(1-(\pi_0^{1/d} u_1)^d)\bar G(w) \\
& + \sum_{k=1}^\infty \sum_j \frac{\pi_0^{1/d}u_{k,j}-\pi_0^{1/d}u_{k+1,j}}{\pi_0^{1/d}u_{k}-\pi_0^{1/d}u_{k+1}}  \bigg((\pi_0^{1/d}u_{k})^d - (\pi_0^{1/d}u_{k+1})^d\bigg) \P\{X_{k,j} > w \}\\
\end{align*}
One can now easily check that $R$ and $\tilde R$ indeed coincide.
\end{proof}

\subsubsection{General Job Sizes}
We further generalize the results given in section \ref{sec:SQd_PH} to the case of general job sizes. In particular we show the following result :
\begin{theorem} \label{thm:response_SQd_gen}
Let $0<\lambda<\mu$ (with $1/\mu$ the mean of the job size distribution) be arbitrary and $R$ the response time for a memory dependent version of the SQ($d$) policy with an arbitrary job size distribution. Further, let $\tilde R$ denote the response time for the
classic SQ($d$) policy with the same job size distribution and arrival rate $\lambda \pi_0^{1/d}$, then $R$ and $\tilde R$ have the same distribution.
\end{theorem}
\begin{proof}
Let us denote by $x_{k}(t,w)$ resp.~$y_k(t,w)$ the density at which, at time $t$, the queue at the cavity has exactly $k$ jobs and the job at the head of the queue has a remaining size exactly equal to $w$ for the memory dependent scheme resp.~the memory independent scheme. Associated to these values, we denote $u_k(t)=\int_0^\infty \sum_{\ell \geq k} x_\ell(t,w) \, dw$ and $v_k(t) = \int_0^\infty \sum_{\ell \geq k} y_\ell (t,w) \, dw$. Furthermore let $x_{k}(w), y_{k}(w)$ and $u_{k}, v_k$ denote the limit of $t\rightarrow\infty$ for these values. We first show that $x_{k}(w)=\pi_0^{-1/d} \cdot y_{k}(w)$ (and consequently also $u_k=\pi_0^{-1/d} \cdot v_k$).

Let us first consider $x_{k}(t,w)$ for $k \geq 2$. Analogously to the proof for exponential and Phase Type job sizes, we obtain:
\begin{align*}
x_k(t+\Delta, w)
&= x_k(t,w+\Delta) - \lambda d x_k(t,w+\Delta) \int_0^\Delta \pi_0(t+\delta) \cdot \sum_{j=0}^{d-1} \frac{1}{j+1} \binom{d-1}{j}\\
& (u_k(t+\delta)-u_{k+1}(t+\delta))^j u_{k+1}(t+\delta)^{d-j-1} \, d \delta + \lambda d x_{k-1}(t, w+\Delta)\\
& \int_0^\Delta \pi_0(t+\delta) \sum_{j=0}^{d-1} \frac{1}{j+1} \binom{d-1}{j} (u_{k-1}(t+\delta)-u_{k}(t+\delta))^j\\
& u_{k+1}(t+\delta)^{d-j-1} \, d \delta + \int_0^\Delta x_{k+1}(t,\delta) g(w+\Delta-\delta) \, d \delta + o(\Delta).
\end{align*}
Subtracting $x_k(t,w)$ on both sides, dividing both sides by $\Delta$ and taking the limit $\Delta \rightarrow 0^+$ we obtain the following system of IDEs:
\begin{align*}
\frac{\partial x_k(t,w)}{\partial t} - \frac{\partial x_k(t,w)}{\partial w}
&=-\lambda \pi_0(t) \frac{x_k(t,w)}{x_k(t)} (u_k(t)^d-u_{k+1}(t)^d)\\
& +\lambda \pi_0(t) \frac{x_{k-1}(t,w)}{x_{k-1}(t)} (u_{k-1}(t)^d - u_k(t)^d) + x_{k+1}(0^+) g(w).
\end{align*}
Taking the limit of $t\rightarrow \infty$ we obtain:
\begin{align}
x_k'(w)&=\lambda \pi_0 \frac{x_k(w)}{x_k}(u_k^d-u_{k+1}^d)-\lambda \pi_0 \frac{x_{k-1}(w)}{x_{k-1}} (u_{k-1}^d-u_k^d)  - x_{k+1} (0^+) g(w). \label{eq:mem_gen_sizes}
\end{align}
A differential equation for the system without memory can be
 inferred from \eqref{eq:mem_gen_sizes} by setting $\pi_0=1$ and replacing $\lambda$
 by $\lambda \pi_0^{1/d}$:
\begin{align*}
y_k'(w)&= \lambda \pi_0^{1/d} \frac{y_k(w)}{y_k} (v_k^d - v_{k+1}^d) - \lambda \pi_0^{1/d} \frac{y_{k-1}(w)}{y_{k-1}} (v_{k-1}^d-v_k^d) - y_{k+1}(0^+) g(w).
\end{align*}
Multiplying both sides by $\pi_0^{-1/d}$, we find that $y_k$ satisfies the following (for $k\geq 2$):
\begin{align*}
(\pi_0^{-1/d} y_k(w))'
&= \lambda \pi_0 \frac{\pi_0^{-1/d} y_k(w)}{\pi_0^{-1/d} y_k} ((\pi_0^{-1/d}v_k)^d-(\pi_0^{-1/d} v_k)^d) - \lambda \pi_0 \frac{\pi_0^{-1/d} y_{k-1}(w)}{\pi_0^{-1/d}y_{k-1}}\\
& ((\pi_0^{-1/d}v_{k-1})^d 
-(\pi_0^{-1/d}v_k)^d)- (\pi_0^{-1/d} y_{k+1}(0^+)) g(w),
\end{align*}
which is identical to \eqref{eq:mem_gen_sizes} if we replace $x_k$ with $\pi_0^{-1/d} y_k$. 

It remains to look at the case $k=1$. For this case, the arrivals we need to consider are those which occur when the queue at the cavity is empty. Therefore, we need to consider two types of arrivals: those which occur due to the fact that the queue at the cavity is in the memory and those which occur due to the queue at the cavity being selected by the SQ($d$) policy. For the arrivals incurred by the memory we find :
\begin{align*}
&\lim_{t\rightarrow\infty} \lim_{\Delta \rightarrow 0^+} \lambda \frac{1}{\Delta} \int_0^\Delta (1-\pi_0(t+\delta)) g(w+\Delta-\delta) \, d\delta + \frac{o(\Delta)}{\Delta} = \lambda (1-\pi_0) g(w).
\end{align*}
The arrivals incurred from the SQ($d$) policy are similar to the case $k\geq 2$, we obtain that $x_1'(w)$ satisfies:
\begin{align}
x_1'(w)
&= \lambda \pi_0 \frac{x_1(w)}{x_1} (u_1^d - u_2^d) - \lambda \pi_0 (1-u_1^d) g(w)- x_2(0^+) g(w) - \lambda (1-\pi_0) g(w). \label{eq:x1_gensizes}
\end{align}
For the system without memory we replace $\pi_0$ by $1$ and $\lambda$ by
$\lambda \pi_0^{1/d}$. If we then multiply both sides by $\pi_0^{-1/d}$, we obtain:
\begin{align}
(\pi_0^{-1/d} &y_1(w))'
= \lambda \pi_0 \frac{\pi_0^{-1/d} y_1(w)}{\pi_0^{-1/d} y_1} ((\pi_0^{-1/d} v_1)^d-(\pi_0^{-1/d} v_2)^d)\nonumber\\
&-\lambda  g(w) + \lambda \pi_0 g(w) (\pi_0^{-1/d} v_1)^d - (\pi_0^{-1/d}y_2(0^+)) g(w). \label{eq:y1_gensizes}
\end{align}
It is not hard to see that \eqref{eq:x1_gensizes} and \eqref{eq:y1_gensizes} are equivalent (with $x_k$ replaced by $\pi_0^{-1/d} y_k$). This shows that we indeed have $x_k(w)=\pi_0^{-1/d} y_k(w)$ for all $k\geq 1$ and $w\geq 0$.

For the response time distribution of the system with memory, we obtain:
\begin{align*}
\bar F_R(w)
&= (1-\pi_0) \bar G(w)\\
& + \pi_0 (1-x_0^d) \bar G(w)+ \pi_0 \sum_{k=1}^\infty \int_0^w \frac{x_k(s)}{x_k} (u_k^d - u_{k+1}^d) \mathbb{P}\{ \mathcal{G}^{*k}  > w-s \} \, ds\\
&= (1- (\pi_0^{1/d} x_0)^d) \bar G(w)\\
& + \sum_{k=1}^\infty \int_0^w \frac{x_k(s)}{x_k} ((\pi_0^{1/d} u_k)^d-(\pi_0^{1/d} u_{k+1})^d) \mathbb{P}\{\mathcal{G}^{*k} > w-s \} \, ds.
\end{align*}
Analogously one can compute $\bar F_{\tilde R}$ to complete the proof.
\end{proof}
\begin{remark}
When $d=1$ in Theorem \ref{thm:response_SQd_gen}, the system without memory reduces to an 
ordinary M/G/1 queue with arrival rate $\lambda \pi_0^{1/d}$ for which many results exist. 
In particular, we find from the Pollaczek-Khinchin formula that the following holds:
\begin{equation}\label{eq:JIQ}
 R^*(w) = \frac{(1-\pi_0^{1/d})\rho \mathcal{G}^*(w)w}{\pi_0^{1/d}\lambda \mathcal{G}^*(w)+w- \pi_0^{1/d} \lambda} 
\end{equation}
with $R^*$  and $\mathcal{G}^*$ the Laplace transform of $R$ and $\mathcal{G}$,
respectively. Using the ISM scheme presented in Section \ref{sec:example_JIQ}, 
this allows one to analyze the JIQ policy with finite memory by plugging $\pi_0=\frac{1-(1-\rho)^{\frac{1}{A+1}}}{\rho}$ into \eqref{eq:JIQ} (see also Proposition \ref{prop:pi_0_JIQ_SQd_gen}).
\end{remark}

Using the ideas in Theorem \ref{thm:response_SQd_gen} we are able to show that $u_1=\rho$ holds:
\begin{proposition} \label{prop:u1_eq_rho_SQd_gen}
For the memory dependent SQ($d$) policy with general job sizes we have $u_1 = \rho$.
\end{proposition}
\begin{proof}
We use the same notation as in the proof of Theorem \ref{thm:response_SQd_gen}. Furthermore, we denote $\tilde x_k(w) = \int_w^\infty x_k(u) \, du$. We now wish to show that $u_1 = \sum_{k=1}^\infty \int_0^\infty x_k(w) \, dw = \rho$.

Integrating \eqref{eq:x1_gensizes} from $w$ to $\infty$, we find:
\begin{align}
x_1(w) &= - \lambda \pi_0 \frac{\tilde x_1 (w)}{x_1} (u_1^d - u_2^d)\nonumber\\
& + \lambda \pi_0 (1-u_1^d) \bar G(w) + x_2(0^+) \bar G(w) + \lambda (1-\pi_0) \bar G(w). \label{eq:x1w_proof_u1}
\end{align}
For $k \geq 2$ we find from \eqref{eq:mem_gen_sizes} that (integrate from $w$ to infinity):
\begin{align}
x_k(w)
&= - \lambda \pi_0 \frac{\tilde x_k(w)}{x_k} (u_k^d - u_{k+1}^d) + \lambda \pi_0 \frac{\tilde x_{k-1}(w)}{x_{k-1}} (u_{k-1}^d - u_k^d) + x_{k+1}(0^+) \bar G(w). \label{eq:xkw_proof_u1}
\end{align}
It is now easy to see from taking the sum of \eqref{eq:x1w_proof_u1} and \eqref{eq:xkw_proof_u1} (for all $k\geq 2$) that for any $w$:
\begin{align}
u_1(w) &= \lambda (1-\pi_0 u_1^d) \bar G(w) + u_2(0^+) \bar G(w). \label{eq:u1w}
\end{align}
Integrating this expression from $0$ to infinity, we obtain:
\begin{equation}\label{eq:u1}
u_1 = \left( u_2(0^+) + \lambda (1-\pi_0 u_1^d) \right) \E[G].
\end{equation}

Furthermore, it is not hard to see that we have for any $k \geq 2$:
\begin{align*}
u_k(t+\Delta) &= \int_0^\Delta \big( u_k(t+\delta) - x_k(t+\delta, \Delta - \delta) \big) d\delta  \\
&+ \lambda \int_0^\Delta \pi_0(t+\delta) \left(u_{k-1}(t+\delta)^d - u_k(t+\delta)^d \right) \, d\delta + o(\Delta),
\end{align*}
$$
u_k'(t) = -x_k(t,0^+) + \lambda \pi_0(t) (u_{k-1}^d(t)- u_k(t)^d),
$$
letting $t\rightarrow \infty$ this leads to:
$$
0 = - x_k(0^+) + \lambda \pi_0 (u_{k-1} ^ d - u_k^d).
$$
Taking the sum of these equations for $k \geq 2$ we obtain:
$$
u_2(0^+) = \lambda \pi_0 u_1^d.
$$
Using this allows us to conclude that $u_1 = \lambda \E[G] = \rho$ from \eqref{eq:u1}.
\end{proof}

In the following Proposition, we obtain $\pi_0$ for the ISM memory scheme presented in Section \ref{sec:example_JIQ}.
\begin{proposition} \label{prop:pi_0_JIQ_SQd_gen}
For the SQ($d$) policy with general job sizes and the ISM memory scheme presented in Section \ref{sec:example_JIQ} we have 
$$
\pi_0=\frac{1-(1-\rho^d)^{\frac{1}{A+1}}}{\rho^d}.
$$
\end{proposition}
\begin{proof}
We use the same notation as in the proof of Proposition \ref{prop:u1_eq_rho_SQd_gen}. 
The rate at which servers send probes is equal to $x_1(0^+)$ (which is equal to the rate at which servers become idle). Therefore, the memory state evolves as a birth-death process with  birth rate $x_1(0^+)$ and death rate $\lambda$. From taking the limit $w\rightarrow 0^+$ in \eqref{eq:u1w}, we find that $x_1(0^+)= \lambda (1- \pi_0 u_1^d)$.

We consequently find that due to the birth-death structure:
\[
\pi_0 = \frac{1}{\sum_{i=0}^A (1-\pi_0 \rho^d)^i}=\frac{\pi_0 \rho^d}{1- (1-\pi_0 \rho^d)^{A+1}}.
\]
From this one easily completes the proof.
\end{proof}
In particular, Proposition \ref{prop:pi_0_JIQ_SQd_gen} holds for $d=1$, which provides a closed form of $\pi_0$ for JIQ with finite memory size.

\subsection{LL($d$)}
For LL($d$), we again start by describing the transient regime (the proof is similar to the one presented in \cite{hellemans2018power}).
\begin{proposition}\label{th:PIDE}
The density of the cavity process associated to the memory dependent LL($d$) policy satisfies the following Partial Integro Differential Equations (PIDEs):
\begin{align}
\frac{\partial f(t,w)}{\partial t} - \frac{\partial f(t,w)}{\partial w} &= \lambda d \pi_0(t)\int_0^w f(t,u) \bar F(t,u)^{d-1} g(w-u) du \nonumber\\
& + \lambda \pi_0(t) (	1-\bar F(t,0)^d) g(w) - \lambda d \pi_0(t) f(t,w)\bar F(t,w)^{d-1}\nonumber \\
& + \lambda (1-\pi_0(t)) g(w) \label{eq:PIDE1}\\
&\frac{\partial \bar F(t,0)}{\partial t} = -f(t,0^+) + \lambda \pi_0(t) (1-\bar F(t,w)^d) + \lambda (1-\pi_0(t)),\label{eq:PIDE2}
\end{align}
for $w>0$, where $f(x,z^+) = \lim_{y \downarrow z} f(x,y)$.
\end{proposition}
\begin{proof}
Assume $w > 0$ and let $w > \Delta > 0$ be arbitrary. As for SQ($d$), we write:
\begin{equation}\label{eq:f1}
f(t+\Delta, w) = Q_{1,w} + Q_{2,w} + Q_{3,w}.
\end{equation}
For $Q_{1,w}$ we consider the case where no arrivals occur in the interval $[t,t+\Delta]$: if the cavity queue at time $t$ has a workload exactly equal to $w + \Delta$ and receives no arrivals 
in $[t,t+\Delta]$, it has a workload equal to $w$ at time $t+\Delta$. Therefore we find:
\begin{align*}
Q_{1,w} &= f(t, w + \Delta) - \lambda d \left( \int_{0}^{\Delta} \pi_0(t+\delta)
f(t+\delta,w+\Delta - \delta) d\delta \right) + o(\Delta).
\end{align*}
For $Q_{2,w}$ we consider the case where a single arrival occurs when the queue 
at the cavity is busy: in this case at some time $t+\delta, \delta \in [0,\Delta]$ an arrival of size $w+\Delta-u$ occurs, while the queue at the cavity has workload $u-\delta$ 
for some $u \in (\delta,w+\Delta]$. This arrival only joins the queue at the cavity if the
other $d-1$ queues have a workload that exceeds $u-\delta$, hence we find:
\begin{align*}
Q_{2,w} &= \lambda d  \int_{0}^{\Delta} \pi_0(t+\delta)\\
& \int_{u=\delta}^{w+\Delta} f(t+\delta, u-\delta) \bar F(t+\delta, u-\delta)^{d-1} g(w+\Delta-u)dud\delta  + o(\Delta).
\end{align*}
Finally a single arrival may occur when the cavity queue is empty: in this case a job 
of size $w + \Delta - \delta$ arrives at time $t+\delta$ for some $\delta \in [0,\Delta]$. Hence,
\begin{align*}
Q_{3,w} &= \lambda d \int_{0}^{\Delta} \pi_0(t+\delta) \frac{(1-\bar F(t+\delta,0)^d)}{d} g(w+\Delta-\delta)d\delta\\
&+\lambda \int_{0}^\Delta \frac{1-\pi_0(t+\delta)}{1-\bar F(t+\delta,0)} (1-\bar F(t+\delta,0)) g(w+\Delta-\delta) d\delta + o(\Delta).
\end{align*}
By subtracting $f(t,w+\Delta)$, dividing by $\Delta$ and letting  $\Delta$ decrease to zero, we find \eqref{eq:PIDE1} from \eqref{eq:f1}.

We still require an equation for $F(t,0)$, the probability that the server is idle.
A server may be idle at time $t+\Delta$ by remaining idle in $[t,t+\Delta]$ or by
having a workload equal to $\Delta - \delta, \delta < \Delta$ at time $t + \delta$. 
We therefore find:
\begin{align*}
&F(t+\Delta, 0) = F(t,0) - \lambda d \int_{0}^{\Delta} \pi_0(t+\delta) 
\frac{(1-\bar F(t+\delta,0)^d)}{d}\, d\delta\\
& -\lambda \int_0^\Delta \frac{1-\pi_0(t+\delta)}{1-\bar F(t+\delta,0)} (1-\bar F(t+\delta,0)) \, d\delta + \int_{0}^{\Delta} f(t+\delta, \Delta - \delta)\, d\delta + o(\Delta),
\end{align*}
subtracting $F(t,0)$, dividing by $\Delta$ and letting $\Delta$ tend to zero yields \eqref{eq:PIDE2} after multiplying both sides by $(-1)$.
\end{proof}

This result readily provides us with the equilibrium workload distribution for the LL($d$) policy
with memory:
\begin{theorem}
The ccdf of the equilibrium workload distribution for the cavity process associated to an LL($d$)
 policy with memory satisfies the following IDE:
\begin{equation}
\bar F'(w)=-\lambda \left[ \bar G(w) + \pi_0 \cdot \left( -\bar F(w)^d + \int_0^w \bar F(u)^d g(w-u) \, du \right) \right]. \label{eq:IDE_LLd}
\end{equation}
with boundary condition $\bar F(0)=\rho$. Equivalently we have:
\begin{equation}
\bar F(w)=\rho - \lambda \int_0^w (1-\pi_0 \bar F(u)^d) \bar G(w-u) \, du. \label{eq:FPE_LLd}
\end{equation}
with $\pi_0$ the probability that the memory is empty.
\end{theorem}
\begin{proof}
To show this result, one first lets $t\rightarrow \infty$ in (\ref{eq:PIDE1}-\ref{eq:PIDE2}), this way we remove the $\frac{\partial f(t,w)}{\partial t}$ and $\frac{\partial \bar F(t,0)}{\partial t}$. One then integrates \eqref{eq:PIDE1} once and uses \eqref{eq:PIDE2} as a boundary condition. Using Fubini, simple integration techniques and the fact that $f(w)=-\bar F'(w)$ we obtain \eqref{eq:IDE_LLd}. The last equality \eqref{eq:FPE_LLd} can be shown by integrating once more and applying Fubini's theorem.
\end{proof}

We can rewrite \eqref{eq:FPE_LLd} as
$$
\pi_0^{1/d} \bar F(w)
=
\E[G] (\lambda \pi_0^{1/d}) - (\lambda \pi_0^{1/d}) 
\int_0^w (1-(\pi_0^{1/d}\bar F(u))^d) \bar G(w-u) \, du.
$$
Comparing this expression with \eqref{eq:FbarW_classicSQd_IDE}, we note that $\bar F(w)$
in a system with memory is equal to the same probability in a system without memory with
arrival rate $\lambda \pi_0^{1/d}$ divided by $\pi_0^{1/d}$. Due to \eqref{eq:FbarW_classicSQd_closed}
we therefore have the following corollary:

\begin{corollary}
The equilibrium workload of the queue at the cavity of an LL($d$) system with memory and
exponential job sizes is given by 
\begin{equation}
\bar F(w)= (\rho \pi_0 + (\rho^{1-d} - \rho \pi_0) e^{(d-1)w})^{\frac{1}{1-d}} \label{eq:closed_LLd}
\end{equation}
\end{corollary}
We are now able to show our main result for a memory dependent LL($d$) policy:
\begin{theorem} \label{thm:response_LLd}
Let $0 < \rho = \lambda \E[G] < 1$  be arbitrary and $R$ the response time of the memory dependent LL($d$) policy with mean job size $\E[G]$ and arrival rate $\lambda$. Further, let $\tilde R$ denote the response time for the same system without memory, but with arrival rate $\lambda \pi_0^{1/d}$, then $R$ and $\tilde R$ have the same distribution
\end{theorem}
\begin{proof}
Let $\bar F(w)$ and $\bar H(w)$ be the ccdf of the workload for the system with and without
memory, respectively. We have $\bar F(w)\pi_0^{1/d}=\bar H(w)$ which yields:
\begin{align*}
\bar F_{R}(w)
&= (1-\pi_0) \bar G(w) + \pi_0 \bigg[ \int_0^w \bar F(w-u)^d g(u) \, du  + \bar G(w)  \bigg]\\
&= \bar G(w) + \int_0^w \bar H(w-u)^d g(u) \, du,
\end{align*}
which can easily be seen to be equal to $\bar F_{\tilde R}(w)$.
\end{proof}

By using the results in this section, one can easily generalise many of the results presented in \cite{hellemans2018power} including an analytical proof that LL($d$) outperforms SQ($d$) and closed form solutions for the response time distribution, mean response time and mean workload.

\begin{proposition}\label{prop:pi_0_JIQ_LLd}
For the LL($d$) policy with the ISM memory scheme presented in Section \ref{sec:example_JIQ} we have $$\pi_0=\frac{1-(1-\rho^d)^{\frac{1}{A+1}}}{\rho^d}$$ for any job size distribution.
\end{proposition}
\begin{proof}
The rate at which servers send probes is equal to $f(0)=-\bar F'(0)$ and it follows from \eqref{eq:IDE_LLd} that $f(0) = \lambda(1-\pi_0 \rho^d)$. The memory state therefore evolves as a birth-death process
with birth rate $\lambda(1-\pi_0 \rho^d)$ and death rate $\lambda$. The remainder of the proof
is therefore identical to the proof of Proposition \ref{prop:pi_0_JIQ_SQd_gen}.
\end{proof}

\begin{center}
\begin{table*}[]
\begin{tabular}{c|ccccc}
      Setup  & $N=10$ & $N=20$ & $N=50$ & $N=100$ & $N=200$ \\ \hline
 $1$ & 1.8839 & 1.5363 & 1.3556 & 1.3059  & 1.2832       \\
 $2$ & 1.4533 & 1.3119 & 1.2313 & 1.2045  & 1.1926       \\
 $3$ & 1.5906 & 1.3860 & 1.2787 & 1.2399  & 1.2215       \\
 $4$ & 1.9086 & 1.3981 & 1.1643 & 1.1158  & 1.0999       \\
 $5$ & 2.3918 & 2.0132 & 1.8200 & 1.7733  & 1.7407       \\
 $6$ & 1.7583 & 1.5920 & 1.4943 & 1.4578  & 1.4404       \\
$7$ & 2.0504 & 1.8040 & 1.6643 & 1.6161  & 1.5901       \\
 $8$ & 2.2790 & 1.5924 & 1.2950 & 1.2352  & 1.2186 
\end{tabular}
\begin{tabular}{c|cccc}
      Setup  & $N=500$ & $N=1000$ & $N=3000$ & Cavity Method \\ \hline
 $1$  & 1.2683  & 1.2638   & 1.2574   & 1.2583       \\
 $2$  & 1.1836  & 1.1810   & 1.1794   & 1.1787       \\
 $3$  & 1.2110  & 1.2097   & 1.2068   & 1.2058       \\
 $4$  & 1.0928  & 1.0921   & 1.0896   & 1.0888       \\
 $5$  & 1.7252  & 1.7178   & 1.7146   & 1.7138       \\
 $6$  & 1.4314  & 1.4304   & 1.4257   & 1.4256       \\
 $7$  & 1.5753  & 1.5736   & 1.5667   & 1.5660       \\
 $8$  & 1.2097  & 1.2096   & 1.2070   & 1.2056 
\end{tabular}
\caption{Comparison of mean response time for the finite system and the cavity method.} \label{table:fin_acc}
\end{table*}
\end{center}

\section{Finite System Accuracy}\label{sec:finAcc}

The results presented in Section \ref{sec:analysis} all focused on the cavity process
of SQ($d$) and LL($d$) with memory. In Table \ref{table:fin_acc} we present simulation results which illustrate that the stationary mean response time in a finite stochastic system consisting of $N$
servers converges to the mean response time obtained using the cavity method.
We simulated a system with $N=10,20,50,100,200,500,1000$ and $3000$ servers. The arrival rate equaled $\lambda N$, the runtime was set to $10^6/N$ and we used a warm-up period equal to a third of the runtime. Job sizes have mean one and are either exponential or hyperexponential with balanced means and a Squared Coefficient of Variation (SCV) equal to $2$ or $3$. 

The following $8$ arbitrarily chosen settings have been considered:
\newline
\textbf{Setup $1$} : LL($4$), $\lambda =0.9$, exponential job sizes and the IP memory scheme.
\newline
\textbf{Setup $2$} : LL($3$), $\lambda=0.8$, exponential job sizes and the CP memory scheme
(meaning memory is of infinite size).
\newline
\textbf{Setup $3$} : LL($3$), $\lambda=0.8$, hyperexponential job sizes with SCV equal to $2$ and BCP memory scheme with $A=5$.
\newline
\textbf{Setup $4$} : LL($2$), $\lambda=0.85$, hyperexponential job sizes with SCV equal to $3$ and the ISM memory scheme with $A=10$.
\newline
Setups $5$ through $8$ are the same as $1$ through $4$, but using SQ($d$) rather than LL($d$).
In all cases the mean response time appears to converge towards the response time of the cavity method.
Note that in the last two setups we are considering SQ($d$) with memory and hyperexponential
job sizes. In this case the response time of the cavity method is simply computed as the
response time in the same system without memory, but with arrival rate $\lambda \pi_0^{1/d}$. 

\begin{figure*}[t]
\begin{subfigure}{.5\textwidth}
\centering
\captionsetup{width=.8\linewidth}
\includegraphics[width=0.9\textwidth]{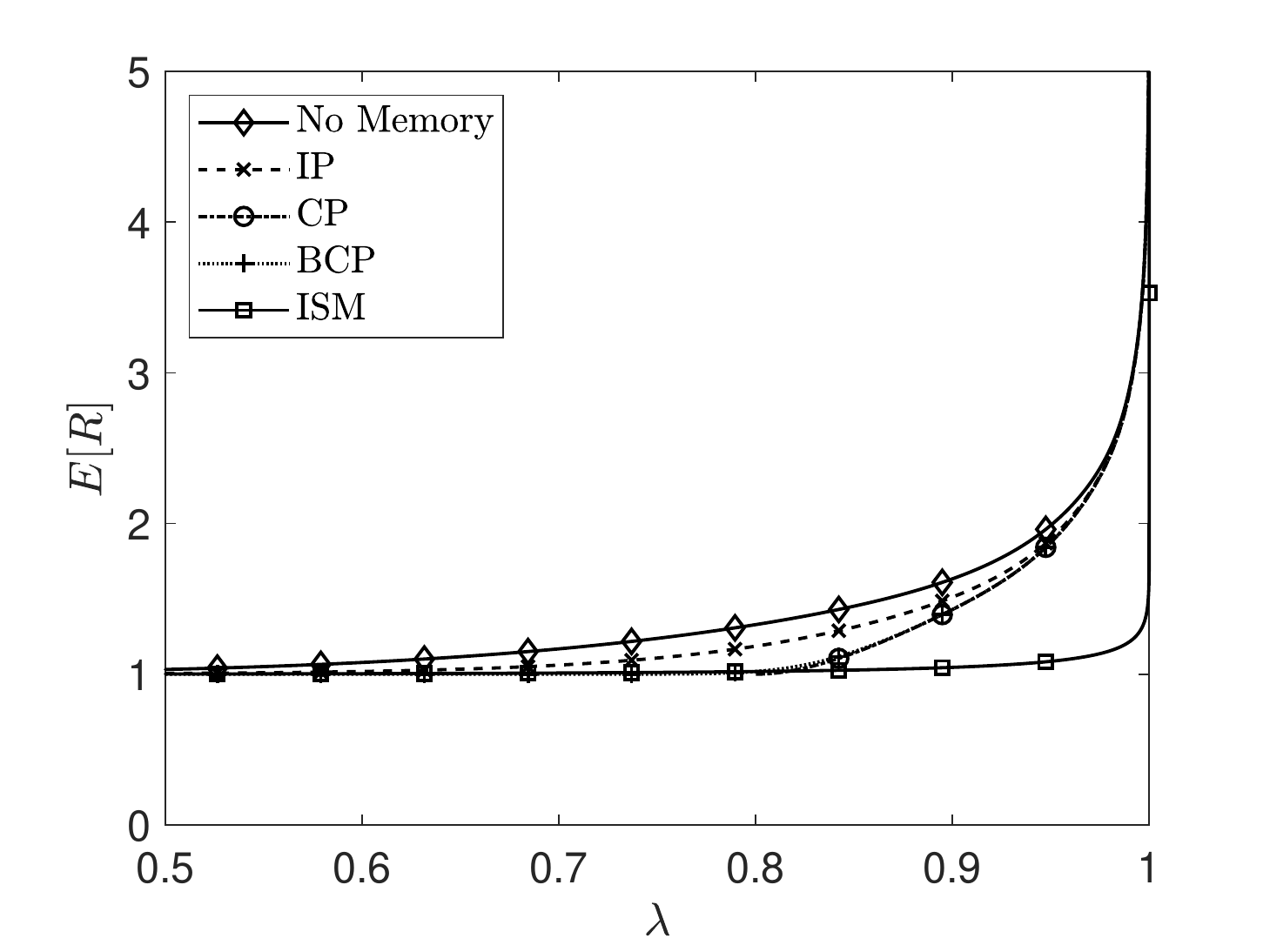}
\subcaption{Mean response time.}
\label{fig:num_exp_ER}
\end{subfigure}
\begin{subfigure}{.5\textwidth}
\centering
\captionsetup{width=.8\linewidth}
\includegraphics[width=0.9\textwidth]{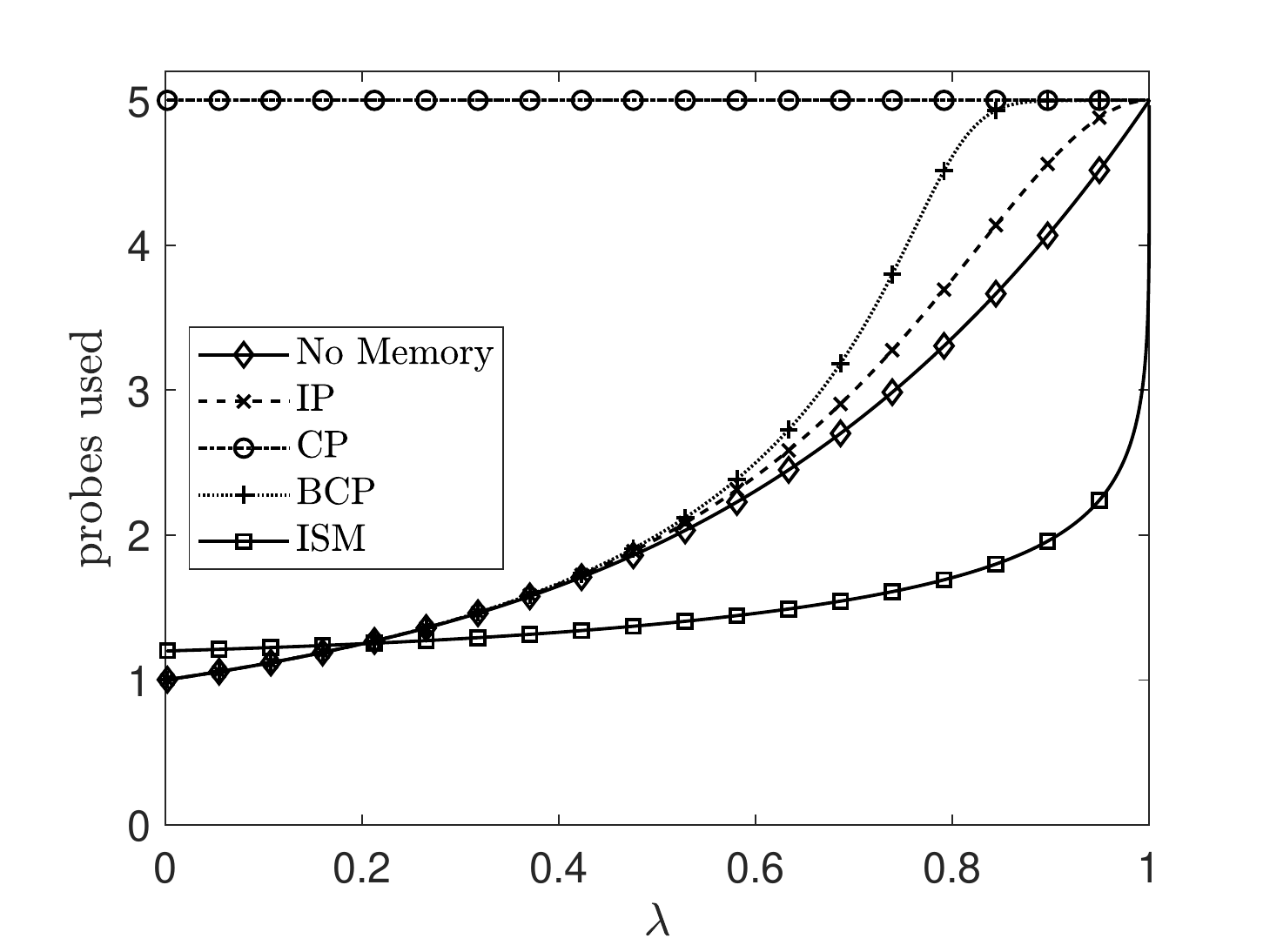}
\subcaption{Number of probes used per arrival.}
\label{fig:num_probes}
\end{subfigure}
\begin{subfigure}{.5\textwidth}
\centering
\captionsetup{width=.8\linewidth}
\includegraphics[width=0.9\textwidth]{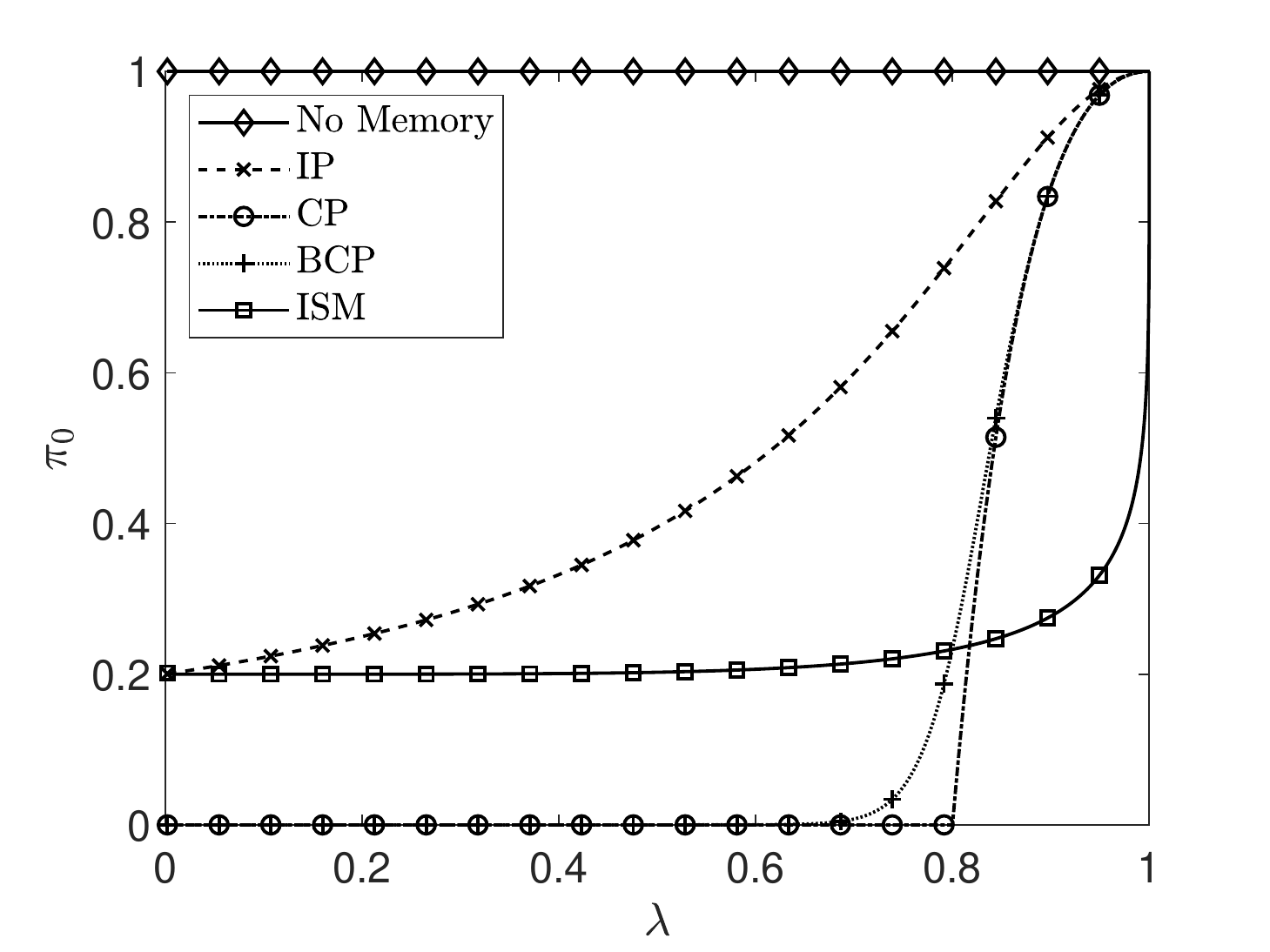}
\subcaption{Prob.~of having empty memory.}
\label{fig:num_exp_pi0}
\end{subfigure}

\caption{Performance of the different memory schemes for SQ($5$) with exponential job sizes with mean one.}
\label{fig:num_exp_SQd}
\end{figure*}

\section{Numerical Example} \label{sec:num_examples}

In this section we briefly demonstrate the type of numerical results that can be obtained using
our findings. This section is not intended as a detailed comparison of the different memory
schemes presented in Section \ref{sec:examples}. 

Figure \ref{fig:num_exp_SQd} focuses on the SQ($5$) policy with exponential job sizes with mean  one and a memory size $A$ of $4$ (except for CP). For the BCP and ISM memory schemes the dispatcher 
is assumed to send its $d$ probes one at a time (if memory is empty upon a job arrival) 
and stops probing as soon as an idle server is found. This is also the case for the setting
without memory (labeled {\it No memory}). For the CP memory scheme we assume that the dispatcher
has infinite memory. 
We plot the mean response times, the probabity of having empty memory $\pi_0$ and
the average number of probes/messages used per job arrival.

In Figure \ref{fig:num_exp_ER} we see that the mean response time is nearly optimal for all schemes
when the load is low (say below $0.5$). For higher loads we see that the ISM scheme is
the best, followed by the CP/BCP, IP and No Memory scheme. The ISM scheme is especially powerful
when the load is close to one as all the other schemes use probing and probes are highly unlikely
to locate an idle server. The results of CP and BCP are very close to each other, which indicates that
a very small amount of memory may suffice.  

In Figure \ref{fig:num_probes} we depict the average number of probes that each of these memory schemes use. If we look at the results for the No Memory, CP/BCP and IP scheme, we see that the schemes
that achieved a lower mean response time use more probes. In this particular case the BCP scheme 
may appear to be superior to CP as it has a similar response time and uses far less probes, but keep
in mind that probes are transmitted one at a time by BCP, while CP can transmit the $d$ probes at
once (which is faster). Looking at both the mean response time and number of probes/messages used, the ISM scheme is clearly best in this case. 


In Figure \ref{fig:num_exp_pi0} we look at the probability of having an empty memory when a job
arrives.  For the IP scheme and a load $\lambda \approx 0$, the dispatcher almost always discovers 
$5$ idle servers and therefore $\pi_0$ is close to $1/5$. For (B)CP we note that as long
as the load is sufficiently low (that is, $5 < \frac{1}{1-\lambda}$ or equivalently $\lambda<4/5$), we have $\pi_0\approx 0$, but for larger $\lambda$ values it sharply increases to one. When $\lambda \approx 4/5$ we also see the most significant gain in response time for (B)CP (see Figure \ref{fig:num_exp_ER}). For the ISM memory scheme, we observe that when $\lambda$ is sufficiently small: 
$$
\pi_0\approx \frac{1}{5} = \frac{1}{A+1} = \lim_{\rho \rightarrow 0^+} \frac{1-(1-\rho^d)^{\frac{1}{A+1}}}{\rho^d},
$$
which is independent of $d$. Only when $\lambda$ is close to one, $\pi_0$ starts a very steep climb to one.
\section{Mean Field Limit Under Heavy Traffic}\label{sec:heavy}
Throughout this section and Section \ref{sec:low}, we assume job sizes are exponential with mean equal to one. The assumption that the mean equals one is merely a technicality to ease notation. Our goal is to compute the limit:
\begin{equation}\label{eq:heavy_traffic_mem}
\lim_{\lambda \rightarrow 1^-} -\frac{\E[R_\lambda]}{\log(1-\lambda)},
\end{equation}
where $R_{\lambda}$ denotes the response time for some SQ($d$)/LL($d$) memory based load balancing policy. To this end, we employ the framework developed in \cite{hellemans2020heavy}. Note that this limit gives an indication of the performance of the load balancing policy under a high load. Moreover, it is easy to see that this limit remains unchanged if we swap the mean response time by either the mean waiting time or the mean queue length/workload.
To emphasize that $\pi_0$ depends on $\lambda$, we denote $\pi_0$ as $\pi_0(\lambda)$
in this section. Define
\begin{equation}\label{eq:T_lam_mem}
T_{\lambda}(x) = \lambda \pi_0(\lambda) x^d,
\end{equation}
and note that $(u_k)_k$ for SQ($d$) resp.~$\bar F(w)$ for LL($d$) satisfy the relations $u_{k+1} = T_{\lambda}(u_k)$ resp.~$\bar F'(w)=T_{\lambda}(\bar F(w)) - \bar F(w)$. 
\begin{theorem} \label{thm:heavy_traffic_mem_Dpi0_finite}
For the memory dependent SQ($d$) policy, provided that $\lim_{\lambda \rightarrow 1^-} \pi_0'(\lambda) < \infty$ and $\lim_{\lambda \rightarrow 1^-} \pi_0(\lambda)=1$, we obtain the heavy traffic limit:
\begin{equation}\label{eq:heavy_traffic_mem_SQd}
\lim_{\lambda \rightarrow 1^-} -\frac{\E[R_\lambda^{(SQ)}]}{\log(1-\lambda)}=\frac{1}{\log(d)},
\end{equation}
while for the LL($d$) variant we have:
\begin{equation}\label{eq:heavy_traffic_mem_SQd}
\lim_{\lambda \rightarrow 1^-} -\frac{\E[R_\lambda^{(LL)}]}{\log(1-\lambda)}=\frac{1}{d-1}.
\end{equation}
\begin{proof}
We validate the requirements (a)--(g) of \cite{hellemans2020heavy} from which the result directly follows.
\begin{enumerate}[label=(\alph*), leftmargin=*]
\item In this step we should show there exists some continuous function $u_{\cdot} : \lambda \rightarrow u_\lambda$ such that $u_\lambda \in (1,\infty)$, $T_\lambda(u_\lambda)= u_{\lambda}$ and $\lim_{\lambda \rightarrow 1^-} u_\lambda = 1$. For our model, it is not hard to find an explicit formula for $u_\lambda=T_\lambda(u_\lambda)$, namely :
$$
u_\lambda=(\lambda \pi_0(\lambda))^{1/(1-d)}.
$$
\item One trivially verifies that $T_\lambda(0)=0$, and for any $u\in (0,1)$ we have $T_\lambda(u) < u$ and $\lim_{\lambda \rightarrow 1^-} \frac{T_\lambda(u)}{u} < 1$.
\item We define $h_\lambda(x) = \frac{u_\lambda-T_\lambda(u_\lambda-x)}{x}$, we now verify that $h_\lambda'(x) < 0$ for any $x\in [u_\lambda-1, u_\lambda]$. To this end, we first compute:
\begin{align*}
h_\lambda'(x)
&= \frac{\lambda \pi_0(\lambda) (u_\lambda - x)^d - u_\lambda + \lambda \pi_0(\lambda) d x (u_\lambda - x)^{d-1}}{x^2}.
\end{align*}
In case $d=1$ this expression simplifies to $-\frac{u_\lambda}{x^2} (1-\lambda \pi_0(\lambda)) < 0$. For $d\geq 2$ we compute the derivative of $(x^2 \cdot h_\lambda'(x))$:
$$
(x^2 h_\lambda'(x))' = - \lambda \pi_0(\lambda) d (d-1) x (u_\lambda - x)^{d-2},
$$
which is negative. As one easily verifies that $(x^2 h_\lambda'(x))$ equals zero in $x=0$, this indeed shows $h_\lambda$ is decreasing on $[u_\lambda-1, u_\lambda]$.
\item This is a technicality which is automatically satisfied because $h_\lambda$ is decreasing on $[u_\lambda-1, u_\lambda]$, which we showed in the previous step.
\item For this step we need to compute the value of:
\begin{align}
A
&=\lim_{\lambda \rightarrow 1^-} h_\lambda(u_\lambda - 1)=\lim_{\lambda \rightarrow 1^-} \frac{u_\lambda - \lambda \pi_0(\lambda)}{u_\lambda - 1}=\lim_{\lambda \rightarrow 1^-} \frac{u_\lambda' - \pi_0(\lambda) - \lambda \pi_0'(\lambda)}{u_\lambda'}.\label{eq:proof_heavy_traffic_mem_Dpi0_finite}
\end{align}
It is not hard to see that:
\begin{align*}
u_\lambda'
&= \frac{1}{1-d} (\lambda \pi_0(\lambda))^{-d/(d-1)} \cdot \left(\pi_0(\lambda) + \lambda \pi_0'(\lambda) \right).
\end{align*}
Using the fact that $\lim_{\lambda \rightarrow 1^-} \pi_0(\lambda) = 1$, we obtain (continuing from \eqref{eq:proof_heavy_traffic_mem_Dpi0_finite}):
\begin{align*}
A
&= \lim_{\lambda\rightarrow 1^-} \frac{\frac{1+ \pi_0'(\lambda)}{1-d} - 1 - \pi_0'(\lambda)}{\frac{1+ \pi_0'(\lambda)}{1-d}} = d.
\end{align*}
\item For this this step, we need to compute the value \label{item:probleempunt_heavy_mem}
\begin{align*}
B&=\lim_{\lambda \rightarrow 1^-} \frac{\log(u_\lambda - 1)}{\log(1-\lambda)}= \lim_{\lambda \rightarrow 1^-} - \frac{(1-\lambda) u_\lambda'}{u_\lambda-1},
\end{align*}
at this point, we use the assumption that $\lim_{\lambda\rightarrow 1^-} \pi_0'(\lambda) < \infty$, as this implies that $u_\lambda' < \infty$, allowing us to use l'Hopital only on $\frac{1-\lambda}{u_\lambda-1}$, by which it trivially follows that $B=1$.
\item For the last step, we should verify that $\lim_{\varepsilon \rightarrow 0^+} \lim_{\lambda \rightarrow 1^-} h_\lambda(\varepsilon)=A$. Indeed,
\begin{align*}
\lim_{\varepsilon \rightarrow 0^+} \lim_{\lambda \rightarrow 1^-} h_\lambda (\varepsilon)
= \lim_{\varepsilon\rightarrow 0^+} \frac{1-(1-\varepsilon)^d}{\varepsilon} = d = A.
\end{align*}
\end{enumerate}
\end{proof}
\end{theorem}
It is not hard to see that Theorem \ref{thm:heavy_traffic_mem_Dpi0_finite} applies to all policies described in Section \ref{sec:examples}, except the ISM policy which we discussed in Section \ref{sec:example_JIQ}. Indeed, ISM is the only policy for which $\lim_{\lambda \rightarrow 1^-} \pi_0'(\lambda) = \infty$, see also Figure \ref{fig:num_exp_pi0}. We therefore find the heavy traffic limit to be slightly different in case of ISM.
\begin{theorem}
For the memory dependent SQ($d$) policy with ISM and a memory size equal to $A$, we obtain the heavy traffic limit:
\begin{equation}\label{eq:heavy_traffic_mem_SQd}
\lim_{\lambda \rightarrow 1^-} -\frac{\E[R_\lambda^{(SQ)}]}{\log(1-\lambda)}=\frac{1}{A+1}\frac{1}{\log(d)},
\end{equation}
while for the LL($d$) variant we have:
\begin{equation}\label{eq:heavy_traffic_mem_SQd}
\lim_{\lambda \rightarrow 1^-} -\frac{\E[R_\lambda^{(LL)}]}{\log(1-\lambda)}=\frac{1}{A+1}\frac{1}{d-1}.
\end{equation}
\end{theorem}
\begin{proof}
One can copy the proof of Theorem \ref{thm:heavy_traffic_mem_Dpi0_finite}, except for step \ref{item:probleempunt_heavy_mem}, as it was used that $\lim_{\lambda \rightarrow 1^-} \pi_0'(\lambda) < \infty$, while it is easy to verify that this limit is indeed infinite for ISM. Let us first compute $u_\lambda'$, using \eqref{eq:pi_0_JIQ} we find:
\begin{align*}
u_\lambda'
&= \left[(\lambda \pi_0(\lambda))^{1/(1-d)} \right] '\\
&= \left( \frac{\lambda}{(1-(1-\lambda^d)^{1/(A+1)})^{1/(d-1)}}\right)'\\
&= \frac{(d-1)(A+1) - d \lambda ^d (1-(1-\lambda^d)^{\frac{1}{A+1}})^{-1} (1-\lambda^d)^{\frac{-A}{A+1}}}{(d-1)(A+1)(1-(1-\lambda^d)^{1/(A+1)})^{1/(d-1)}},
\end{align*}
noting that (by a simple application of l'Hopital's rule) we have $\lim_{\lambda \rightarrow 1^-} \frac{1-\lambda}{1-u_{\lambda}}=0$, we obtain:
\begin{align*}
B
&=\lim_{\lambda \rightarrow 1^-} \frac{(1-\lambda) u_\lambda'}{1-u_\lambda}\\
&= -\lim_{\lambda\rightarrow 1^-} \bigg[ \frac{1-\lambda}{1-u_\lambda} \frac{d \lambda^d}{(d-1)(A+1)} \frac{(1-\lambda^d)^{-A/(A+1)}}{(1-(1-\lambda^d)^{1/(A+1)})^{d/(d-1)}} \bigg]\\
&= -\frac{d}{(d-1)(A+1)} \lim_{\lambda\rightarrow 1^-} \left[ \frac{1-\lambda}{(1-\lambda^d)^{A/(A+1)}(1-u_\lambda)} \right]\\
&=-\frac{d}{(d-1)(A+1)} \lim_{\lambda\rightarrow 1^-} \frac{1-\lambda}{1-\lambda^d} \cdot \lim_{\lambda \rightarrow 1^-} \frac{(1-\lambda^d)^{1/(A+1)}}{1-u_\lambda} \\
&= -\frac{1}{(d-1)(A+1)} \cdot \lim_{\lambda \rightarrow 1^-} \frac{(1-\lambda^d)^{1/(A+1)}}{1-
\frac{\lambda}{\left(1-(1-\lambda^d)^{1/(A+1)}\right)^{1/(d-1)}}}.
\end{align*}
For ease of notation, let us define $\xi = (1-\lambda^d)^{1/(A+1)}$. We find that the above simplifies to :
\begin{align*}
B
&= -\frac{1}{(d-1)(A+1)} \lim_{\xi \rightarrow 0^+} \frac{\xi (1-\xi)^{1/(d-1)}}{(1-\xi)^{1/(d-1)} - (1-\xi^{A+1})^{1/d}} = \frac{1}{A+1},
\end{align*}
where the last equality follows from a final use of l'Hopital's rule. Combining this with the proof of Theorem \ref{thm:heavy_traffic_mem_Dpi0_finite}, we may conclude the proof.
\end{proof}
\section{Mean Field Limit Under Low Traffic}\label{sec:low}
In this section, we investigate the system in the low traffic limit rather than the heavy traffic limit. In particular, we investigate the behaviour of the expected waiting time in the low traffic limit, i.e.~$\lim_{\lambda \rightarrow 1^-} (\E[R_\lambda]-1)$. The value of this limit is always zero but we show that in case of exponential job sizes, we are able to obtain a closed form expression for 
$$
\lim_{\lambda \rightarrow 1^-} \frac{\E[R_\lambda^{(1)}]-1}{\E[R_\lambda^{(2)}]-1}.
$$
Here $R_\lambda^{(1)}, R_\lambda^{(2)}$ denote the response time distribution of two different load balancing policies with arrival rate $\lambda$. We take the quotient of the expected waiting times rather than response times as the quotient for the expected response times is trivially one for any 2 load balancing policies. This quantity signifies the quality of a policy under a low arrival rate.
In particular we have the following result:
\begin{proposition}
Let $R_\lambda^{(i)}$ ($i=1,2$) denote the response time for a memory dependent load balancing policy with probability $\pi_0^{(i)}(\lambda)$ of having an empty memory, using either SQ($d_i$) or LL($d_i$). Furthermore, we assume job sizes are exponentially distributed with mean $1$. We find:
\begin{enumerate}
\item If $d_1 < d_2$, we have 
$$
\lim_{\lambda \rightarrow 0^+} \frac{\E[R_\lambda^{(1)}]-1}{\E[R_\lambda^{(2)}]-1}=\infty.
$$
\item If $d_1=d_2=d$ and both policies use the same strategy (either SQ($d$) or LL($d$)), we have:
$$
\lim_{\lambda \rightarrow 0^+} \frac{\E[R_\lambda^{(1)}]-1}{\E[R_\lambda^{(2)}]-1}= \lim_{\lambda\rightarrow 0^+} \frac{\pi_0^{(1)}(\lambda)}{\pi_0^{(2)}(\lambda)}.
$$
\item If $d_1=d_2=d$ and  $R_\lambda^{(1)}$ employs the SQ($d$) policy while $R_\lambda^{(2)}$ uses the LL($d$) policies, we find:
$$
\lim_{\lambda \rightarrow 0^+} \frac{\E[R_\lambda^{(1)}]-1}{\E[R_\lambda^{(2)}]-1}= \lim_{\lambda\rightarrow 0^+} \frac{\pi_0^{(1)}(\lambda)}{d \pi_0^{(2)}(\lambda)}.
$$ \label{low:item3}
\end{enumerate}
\end{proposition}
\begin{proof}
From \eqref{eq:ER_SQd_vanilla} with arrival rate $\lambda \pi_0^{1/d}$ and mean job size equal to $1$, one finds that the expected response time for SQ($d$) is given by:
$$
\E[R_\lambda^{(SQ(d))}]-1=\sum_{n=2}^\infty (\lambda\pi_0^{1/d})^{\frac{d^n-1}{d-1}-1}.
$$
By Theorem \ref{thm:response_SQd}, we find that this expression corresponds to the mean response time of a memory based SQ($d$) policy. Analogously, for LL($d$) and using Theorem \ref{thm:response_LLd}, we obtain that the expected response time for a memory based LL($d$) policy with exponential job sizes of mean one is given by:
\begin{equation} \label{eq:proof_low_LLd}
\E[R_\lambda]-1=\sum_{n=1}^\infty \frac{(\lambda\pi_0^{1/d})^{dn}}{1+n(d-1)}.
\end{equation} In order to compute the sought limits, one only retains the terms with the lowest power of $\lambda$. For example, assume $d_1< d_2$ and we wish to compare LL($d_1$) with LL($d_2$), it follows from \eqref{eq:proof_low_LLd} that:
$$
\lim_{\lambda \rightarrow 0^+} \frac{\E[R_\lambda^{(1)}]-1}{\E[R_\lambda^{(2)}]-1}= \frac{\lambda^{d_1} \pi_0}{\lambda^{d_2} \pi_0} \cdot \frac{1+(d_2-1)}{1+(d_1-1)} = \infty.
$$
As a second example let us consider case \eqref{low:item3}, we find:
\begin{align*}
\lim_{\lambda \rightarrow 0^+} \frac{\E[R_\lambda^{(1)}]-1}{\E[R_\lambda^{(2)}]-1}&= \frac{(\lambda\pi_0^{1/d})^d}{\frac{(\lambda\pi_0^{1/d})^d}{d}} =\lim_{\lambda\rightarrow 0^+} \frac{d\pi_0^{(1)}(\lambda)}{\pi_0^{(2)}(\lambda)}.
\end{align*}
\end{proof}
\begin{remark}
For the methods discussed in Section \ref{sec:examples} we can easily compute the limit $\lim_{\rho \rightarrow 0^+} \pi_0(\rho)$. Indeed, by elementary calculus we find:
\begin{itemize}
\item For IP we have $\lim_{\rho \rightarrow 0^+} \pi_0(\rho)=\frac{1}{d}$.
\item For CP and BCP we have $\lim_{\rho \rightarrow 0^+} \pi_0(\rho)=0$.
\item For ISM with memory size $A$ we have $\lim_{\rho \rightarrow 0^+} \pi_0(\rho)=\frac{1}{A+1}$.
\end{itemize}
In particular, we see that, while ISM is the dominant policy in heavy traffic, it does not perform as well in low traffic. See also Figure \ref{fig:num_exp_pi0} for an example.
\end{remark}

\section{Conclusions and Future Work} \label{sec:conclusions}
In this paper we studied the cavity process of the SQ($d$) and LL($d$) load balancing
policies with memory. The main insight provided was that the response time distribution
of the cavity process with memory is identical to the response time distribution 
of the cavity process of the system without memory if the arrival rate is properly set.
This result holds for a large variety of memory schemes including the ones 
presented in Section \ref{sec:examples}. This insight allowed us to analyse the heavy and low traffic limit. Simulation results were presented which suggest that the cavity process corresponds to the exact limit process as the number of servers tends to infinity.

As future work, it may be possible to prove that the cavity process is the proper limit process. For SQ($d$) with exponential job sizes, one can build upon the framework of \cite{benaim2008class}, whilst for LL($d$) it might be possible to extent the framework in \cite{shneer2020large} to prove the ansatz for general job sizes.

\bibliographystyle{ACM-Reference-Format}
\bibliography{thesis}

\end{document}